\documentclass[12pt,fullpage]{amsart}

\usepackage{bbm}
\usepackage{amsmath}
\usepackage{amsfonts}
\usepackage{mathrsfs}

\usepackage{amssymb}

 \newtheorem{thm}{Theorem}[section]
\newtheorem{rem}[thm]{Remark}
\newtheorem{definition}[thm]{Definition}
\newtheorem{lem}[thm]{Lemma}
 \newtheorem{prop}[thm]{Proposition}
\newtheorem{cor}[thm]{Corollary}
\newtheorem{exm}[thm]{Example}
 \newtheorem{conj}[thm]{Conjecture}
 \newtheorem{prob}[thm]{Problem}
\def\f#1{{\mathbb{F}}_{#1}}

\begin{document}

\title{Explicit Deep Holes of Reed-Solomon Codes}
\author{Jun Zhang}
\address{Jun Zhang, School of Mathematical Sciences, Capital Normal University, Beijing 100048, P.R. China.}
\email{junz@cnu.edu.cn}
\author{Daqing Wan}
\address{Daqing Wan is with Department of Mathematics,
 University of California, Irvine, CA 92697, USA.}
\email{dwan@math.uci.edu}

\thanks{
	The research of Jun Zhang was supported by the National Natural Science Foundation of China under Grant No. 11601350, by Scientific Research Project of Beijing Municipal Education Commission under Grant No. KM201710028001, and by Beijing outstanding talent training program under Grant No. 2014000020124G140. Jun Zhang is supported by Chinese Scholarship Council for his visit at the University of Oklahoma, USA.  The research of Daqing Wan was supported by National Science Foundation.
}
%\author{Jun Zhang\thanks{Jun Zhang is with the School of Mathematical Sciences, Capital Normal University, Beijing 100048, China. Email: junz@cnu.edu.cn},
%\and
%Daqing Wan
%\thanks{Daqing Wan is with the Department of Mathematics, University of California, Irvine, CA 92697, USA. Email: dwan@math.uci.edu}}

%\keywords{Reed-Solomon codes, covering radius, deep holes}
%\date{}
%\maketitle

\begin{abstract}
In this paper, deep holes of Reed-Solomon (RS) codes are studied. A new class of deep holes for generalized affine RS codes is given if the evaluation set satisfies certain combinatorial structure. Three classes of deep holes for projective Reed-Solomon (PRS) codes are constructed explicitly. In particular, deep holes of PRS codes with redundancy three are completely obtained when the characteristic of the finite field is odd. Most (asymptotically of ratio $1$) of the deep holes of PRS codes with redundancy four are also obtained.
\end{abstract}
\maketitle

% \fi
%\IEEEpeerreviewmaketitle

\section{Introduction}
%\subsection{Notations and the Main Results}
Let $\f{q}^n$ be the $n$-dimensional vector space over the finite field $\f{q}$ of $q$ elements with characteristic $p$. For any vector (also, called \emph{word}) $ {x}=(x_1,x_2,\cdots,x_n)\in \f{q}^n$, the \emph{Hamming weight} $\mathrm{Wt}( {x})$ of $ {x}$ is defined to be the number of non-zero coordinates, i.e.,
$$\mathrm{Wt}( {x})=|\left\{i\,|\,1\leqslant i\leqslant n,\,x_i\neq 0\right\}|.$$
For integers $0\leq k\leq n$,
a \emph{linear $[n,k]$ code} $C$ is a $k$-dimensional linear subspace of $\f{q}^n$. The \emph{minimum distance} $d(C)$ of $C$ is the minimum Hamming weight among all non-zero vectors in $C$, i.e.,
$$d(C)=\min\{\mathrm{Wt}( {c})\,|\, {c}\in C\setminus\{ {0}\}\}.$$
A linear $[n,k]$ code $C\subseteq \f{q}^n$ is called an $[n,k,d]$ linear
code if $C$ has minimum distance $d$. A well-known trade-off between
the parameters of a linear $[n,k,d]$ code is the Singleton bound
which states that
$$d\leqslant n-k+1.$$
 An $[n,k,d]$ code is called a \emph{maximum distance separable}
  (MDS) code if $d=n-k+1$. An important class of MDS codes are affine Reed-Solomon codes and projective Reed-Solomon codes, which will be our main object
  of study in this paper.

 Let $C$ be an $[n,k,d]$ linear code over $\f{q}$. The \emph{error distance} of any word $u\in\f{q}^n$ to $C$ is defined to be
$$d(u,C)=\min\{d(u,v)\,|\,v\in C\},$$
where $$d(u,v)=|\{i\,|\,u_{i}\neq v_{i},\,1\le i\le n\}|$$
is the Hamming distance between words $u$ and $v$.
The error distance plays an important role in the decoding of the code. The maximum error distance
\[
   \rho(C)=\max\{d(u,\, C)\,|\,u\in \f{q}^n\}
\]
is called the \emph{covering radius} of $C$.

Covering radius of codes was studied extensively~\cite{CKMS85,CLS86,GS85,HKM78,MCL84,OST99}.
For MDS codes, the covering radius is known to be either $d-1$ or $d-2$~\cite{GK98}. For a general
MDS code, determining the exact covering radius is difficult. We shall see below that
affine Reed-Solomon codes have covering radius $d-1$. In contrast, the covering radius of projective Reed-Solomon codes
is unknown in general but is conjectured to be $d-2$ (except when $q$ is even and $k=2,\,q-2$ in which case the covering radius is $d-1$).
We now recall the definition of affine and projective  Reed-Solomon codes.
\begin{definition}
Fix a subset
$D=\{x_1,\ldots,x_n\}\subseteq \f{q}$, which is
called the evaluation set. Let $v_1,v_2,\cdots,v_n\in\f{q}^*$. For any integer $0<k<n$, the generalized (affine) RS code
$C=RS(D,k)$ of length $n$, dimension $k$ and scale vector $v=(v_1,v_2,\cdots,v_n)$ over
$\f{q}$ is defined to be
\begin{align*}
    RS_v(D,k)=\{(v_1f(x_1),\ldots,v_nf(x_n))\in
\f{q}^n\,|\, f(x)\in \f{q}[x], \deg f(x)\leq k-1\}.
\end{align*}
\end{definition}
It is easy to check that the minimal distance of this code is $n-k+1$, and thus $RS_v(D, k)$ is an MDS code.
For the case $v=(1,1,\cdots,1)$, we write $RS(D,k)$ for short. For any $v\in(\f{q}^*)^n$, the RS code $RS_v(D, k)$
is equivalent to $RS(D,k)$. Thus,  we only consider $RS(D,k)$ throughout this paper. When $D=\f{q}$, we write $RS(q,k)$ for short.

For any word $u\in \f{q}^{n}$, by the Lagrange interpolation, there is a unique
polynomial $f$ of degree $\le n-1$ such that
\[
   u=u_{f}=(f(x_{1}),f(x_{2}),\cdots,f(x_{n})).
\]
Clearly, $u_{f}\in RS(D,k)$ if and only if $\deg(f)\le k-1$. We also say that $u_{f}$ is defined by the polynomial $f(x)$.
One can easily show~(see \cite{LW08}):
for any $k \leqslant \deg(f)\leqslant n-1$, we have the inequality
 \[
   n-\deg(f)\leqslant d(u_{f},RS(D,k))\leqslant n-k.
\]
It follows that if $\deg(f)=k$, then $d(u_{f},RS(D,k))=n-k$.  One deduces the covering radius of all affine RS codes.

\begin{prop}\label{RSC}
The covering radius of all affine RS codes with parameters $[n,k,d=n-k+1]$ is $n-k=d-1$.
\end{prop}

If the distance from a word to the code achieves the covering radius of the code, then the word is called a \emph{deep hole} of the code. Deciding deep holes of a given code is much harder than the covering radius problem, even for affine RS codes. The deep hole problem for affine RS codes was studied in~\cite{CMP11,CM07,KW15,LW08,LZ15,Liao11,WL08,WH12,ZFL13}. %Methods include estimation of character sums, subset sum problem, algebraic geometry...
As noted above, words $u_f$ with $\deg(f)=k$ are deep holes of affine $RS(D, k)$.
Based on numerical computations, Cheng and Murray~\cite{CM07} conjectured that the converse is also true if $D=\f{q}$.

\begin{conj}[\cite{CM07}]\label{deepholeconj} For $2\leq k\leq q-2$ (and $k\not= q-3$ if $q$ is even),
a word $u_f$ is a deep hole of $RS(q, k)$ if and only if $\deg(f)=k$.
\end{conj}

This conjecture remains open, but has been proved
if either $k+1\le p$ by \cite{ZCL16} or $k\geq (q-1)/2$ by \cite{Kaipa17}. In particular, the conjecture is true for prime fields.

In an opposite direction, Wu and Hong~\cite{WH12} found a new class of deep
holes for $RS(\f{q}^*,k)$ working on cyclic codes via Fourier transformation. Note that here the evaluation set $\f{q}^*$ is a proper subset
of $\f{q}$.
Then Zhang, Fu and Liao~\cite{ZFL13}
realized the essence of the new deep holes and generalized it to all affine RS codes $RS(D,k)$ for any proper subset
$D\subset \f{q}$.
\begin{prop}[\cite{ZFL13}]\label{newdeephole12}
For any proper subset $D\subset \f{q}$, the polynomials $$\{{a}{(x-\delta)}^{q-2}+f_{\leq k-1}(x)\,|\,a\neq 0,\,\delta\in\f{q}\setminus D\} $$
define deep holes of RS code $RS(D,k)$, where $f_{\leq k-1}(x)$ runs all polynomials of degree $\leq k-1$.
\end{prop}

The essence of this class of deep holes is the equality
\[
   {(x-\delta)}^{q-2}=\frac{1}{x-\delta}\,\,\mbox{on}\,\,D.
\]
It makes the exponent set of the defining Laurent monomials of the new code
\[
   C=RS(D,k)\oplus \f{q}\cdot u_{{(x-\delta)}^{q-2}}
\]
to form consecutive integers $\{-1,0,1,\cdots,k-1\}$. So the new code $C$
is indeed a generalized RS code with the scale vector $$v=(\frac{1}{x_1-\delta},\frac{1}{x_2-\delta},\cdots,\frac{1}{x_n-\delta}),$$ which is MDS. And hence, the vectors defined by ${(x-\delta)}^{q-2}$ is a deep hole of $RS(D,k)$. Note that the deep holes defined by polynomials of degree $k$ have the same essence, i.e.,
the exponent set $\{0,1,\cdots,k\}$ is consecutive. For exponents consisting of consecutive integers, the full rank submatrices are Vandermonde matrices whose determinant are easy to compute. However, there is no uniform formula of determinant for general exponent sets. This is one reason why it is difficult to find other new deep holes.

For a general evaluation set $D$,
Zhuang, Cheng and Li in~\cite{ZCL16} proved that for prime $q>2$ and $k\geq \lfloor\frac{q-1}{2}\rfloor$,   the
polynomials in the above proposition together with polynomials of degree $k$ define all the deep holes of RS code $RS(D,k)$. And it is very recently proved
by Kaipa \cite{Kaipa17}  that the
same statement is true for all odd prime power $q$ under the same assumption $k\geq \lfloor\frac{q-1}{2}\rfloor$. Note that these results are true
for all subsets $D$ of $\f{q}$ under the assumption that $|D|=n > k\geq \lfloor\frac{q-1}{2}\rfloor$ and $q$ odd.  One might conjecture
it holds for all odd prime power $q$ and for $D$ with $|D|>k>2$. In Section~\ref{newdeephole17}, we will show that this conjecture is false for $k<\lfloor\frac{q-1}{2}\rfloor$.

The aim of this paper is to construct new deep holes of  projective Reed-Solomon codes. This turns out to be significantly more difficult.

\begin{definition}
The projective Reed-Solomon (PRS) code is defined to be
\begin{align*}
    PRS(q+1,k)=\{(f(\alpha_1),\cdots,f(\alpha_q),c_{k-1}(f))\,|\,f(x)\in \f{q}, \deg(f(x))<k\}
\end{align*}
where $\f{q}=\{\alpha_1,\alpha_2,\cdots,\alpha_q=0\}$ and $c_{k-1}(f)$ is the coefficient of the term of degree $k-1$ of $f(x)$.
\end{definition}
In the literature, it is also called doubly extended RS code.
As above, we introduce more notations.
As the first $q$ coordinates of $PRS(q+1,k)$ corresponds to an affine part of the projective line, we represent a vector in $\f{q}^{q+1}$ by $(u_f,v)$ where $u_f\in\f{q}^q$ is defined by a polynomial $f$ of degree $\le q-1$:
\[
u_f=\left(f(\alpha_1),\cdots,f(\alpha_q)\right)
\]
and $v\in \f{q}$ is arbitrary. It is easy to see that $(u_f,v)\in \f{q}^{q+1}$ represents a codeword of $PRS(q+1,k)$ if and only if
\[
\deg(f)\le k-1\,\mbox{and}\,v=c_{k-1}(f).
\]

It turns out affine RS codes are affine parts of the PRS code. Although they look very similar, the PRS code is much more difficult.
For instance, the covering radius of the PRS code is already unknown.
A consequence~\cite{Arne94} of the MDS conjecture %(see Conjecture~\ref{mds})
implies the following covering radius conjecture for projective Reed-Solomon codes.

\begin{conj}\label{cover}
For $2\leq k \leq q-2$, the covering radius of $PRS(q+1,k)$ is $q-k$ except when $q$ is even, $k=2, q-2$ in which case the
covering radius is $q-k+1$.
\end{conj}

It is shown in~\cite{ZW16}\cite{Kaipa17} that
the deep hole Conjecture~\ref{deepholeconj} also implies the above covering radius conjecture.
Since the deep hole conjecture is known to be true if $k+1\leq p$ by \cite{ZCL16} or $q$ is an odd prime power and $k\geq (q-1)/2$
by \cite{Kaipa17}, we deduce that the covering radius conjecture is true if $k+1\leq p$ or if $q$ is an odd prime power and $k\geq (q-1)/2$.
There are many other cases that the covering radius conjecture is known to be true, see \cite{Kaipa17}, including the two simplest non-trivial
cases $k=q-2$ and $k=q-3$ which will be used later in this paper.
The next major problem is the deep hole problem for $PRS(q+1,k)$.

\begin{prob}\label{prob} Let $2\leq k \leq q-2$. Assume that the covering radius conjecture is true.
Construct all deep holes of $PRS(q+1,k)$.
\end{prob}

This problem is very difficult and wide open. As shown in  \cite{Kaipa17}, classifying deep holes of $PRS(q+1,2)$ for $q$ even
is equivalent to the difficult problem of classifying hyperovals of $PG(2, q)$. In~\cite{ZW16}, we obtained a class of deep holes of $PRS(q+1,k)$ which are defined by polynomials of degree $k$, if $\rho(PRS(q+1,k))=q-k$.

\begin{thm}[Deep holes of degree $k$~\cite{ZW16}]\label{deephole:degreek}
	Let $q$ be a prime power and let $k$ be an integer such that $2\leq k\leq q-2$ if $2\nmid q$ or $3\leq k\leq q-3$ if $2\mid q$. Suppose $\rho(PRS(q+1,k))=q-k$.
	The words $\{(u_f,v)\,|\,\deg(f)=k,\,v\in\f{q}\}$ are deep holes of $PRS(q+1,k)$.
\end{thm}

The reason why it is difficult to find deep holes of a given linear code is that it is generally \textbf{NP}-hard to compute the error distance of a given word. In this paper, we give a simple but very effective method, embedding the code $PRS(q+1,k)$ to MDS supercodes, to control the error distance so that two new classes of deep
holes of $PRS(q+1,k)$ are constructed explicitly.
\begin{thm}[Construction from irreducible quadratic polynomials]\label{main:irr2}
	Let $q$ be a prime power and let $k$ be an integer such that $2\leq k\leq q-2$ if $2\nmid q$ or $3\leq k\leq q-3$ if $2\mid q$. Let $C=PRS(q+1,k)$. Suppose $\rho(C)=q-k$. Let $p(x)=x^2+\alpha x+\beta\in \f{q}[x]$ be a monic irreducible quadratic polynomial. For any $a,b\in \f{q}$ which are not zero simultaneously, the rational function $\frac{a+bx}{p(x)}$ defines the following coset of deep holes of the PRS code $PRS(q+1,k)$:
	\[
	\left(\frac{a+b\alpha_1}{p(\alpha_1)},\frac{a+b\alpha_2}{p(\alpha_2)},\cdots,\frac{a+b\alpha_q}{p(\alpha_q)},0\right)+C.
	\]
\end{thm}

	\begin{thm}[Construction from irreducible cubic polynomials]\label{main:irr3}
	Let $q$ be a prime power. Let $C=PRS(q+1,q-3)$. Let $p(x)\in\f{q}[x]$ be a monic irreducible cubic polynomial. Besides the deep holes in the above theorems,
   there are new cosets of deep holes of $C$ of the form
		\[
		 \left(\frac{a+b\alpha_1+c\alpha_1^2}{p(\alpha_1)},\frac{a+b\alpha_2+c\alpha_2^2}{p(\alpha_2)},\cdots,\frac{a+b\alpha_q+c\alpha_q^2}{p(\alpha_q)},0\right)+C
		\]
	for some non-zero polynomial $a+bx+cx^2\in\f{q}[x]$.
	\end{thm}

   In particular, for the case $q$ is odd and $k=q-2$ the covering radius of the code $PRS(q+1,q-2)$ is known to be $q-k=2$. Kaipa~\cite{Kaipa17} gave a classification of the deep holes of $PRS(q+1,q-2)$ using the $PGL(2,q)$ group action on the projective space $PG(2,q)$ of the syndromes. Under the group action, the syndrome space $PG(2,q)$ is divided into $3$ orbits: one is the orbit of syndromes of non-deep-holes, the other two are orbits of syndromes of deep holes. But he did not give the explicit deep holes. In this paper, we will show that if $p(x)$ runs over all the monic irreducible quadratic polynomials, the explicit deep holes in Theorem~\ref{main:irr2} form all the deep holes of $PRS(q+1,q-2)$. The difficulty here is to distinguish the deep holes for different irreducible quadratic polynomials. To overcome the difficulty, we construct a hypergraph related to the distribution of irreducible quadratic polynomials and investigate the properties of the hypergraph in detail.
  	\begin{thm}\label{main:k=q-2}
  	Let $q$ be an odd prime power. When $p(x)$ runs over all the monic irreducible polynomials of degree $2$ over $\f{q}$, the deep holes in Theorem~\ref{main:irr2} are all the deep holes of $PRS(q+1,q-2)$.
  \end{thm}

The rest of this paper is organized as follows. In Section~\ref{newdeephole17}, we give a new class of deep holes of affine RS code $RS(D,k)$ besides the deep holes defined by polynomials of degree $k$ and
the deep holes in Proposition~\ref{newdeephole12}. In Section~\ref{PRSCode}, we prove Theorem~\ref{main:irr2}. We also investigate the relationship between the deep holes in Theorem~\ref{deephole:degreek} and the deep holes in Theorem~\ref{main:irr2}. In Section~\ref{section:k=q-2}, we prove Theorem~\ref{main:k=q-2}. In Section~\ref{section:k=q-3}, we prove Theorem~\ref{main:irr3} and show that the deep holes in Theorems~\ref{deephole:degreek} and~\ref{main:irr2} ($k=q-3$) are contained in the deep holes in Theorem~\ref{main:irr3} where the L-function theory of exponential sums is applied.

\section{New Deep Holes of Generalized RS Codes}\label{newdeephole17}
In this section, we show that if the evaluation set of the generalized RS Code is zero-sum-free (a combinatorial structure, see Definition~\ref{def:sumfreeset}), then new deep holes of the RS code are obtained.

We first review some results about deep holes of RS codes $RS(D,k)$.
Recall that all deep holes of $RS(q, k)$ were conjectured to be defined by polynomials of degree $k$ (Conjecture~\ref{deepholeconj}).
The same conjecture is false for general evaluation set $D$, see~\cite{WH12,ZFL13}.
To support Conjecture~\ref{deepholeconj}, the easiest case is to determine whether a polynomial of degree $k+1$ defines a deep hole of $RS(q, k)$. Li and Wan~\cite{LW08} interpreted it as a subset sum problem (SSP).
\begin{definition}
Let $G$ be a finite abelian group and $D$ a subset of $G$. The $k$-SSP over $D$ consists in determining for any $g\in G$, if there is a subset $S\subset D$ such that $|S|=k$ and $\sum_{s\in S}s=g$. And denote the number of solutions by
\[
  N(k,g,D)=\left|\left\{S\subset D\,\big|\,|S|=k,\,\sum_{s\in S}s=g\right\}\right|.
\]
\end{definition}

\begin{lem}[\cite{LW08}]\label{SSPDeephole}
For any polynomial $f_{\leq k-1}(x)$ of degree $\leq k-1$, the polynomial $x^{k+1}-ax^k+f_{\leq k-1}(x)$ does not define a deep hole of the RS code $RS(D,k)$ if and only if $N(k+1,a,D)>0$.
\end{lem}

For general $D$, solving $k$-SSP is an \textbf{NP}-hard problem. But for the case $D=G=\f{q}$, the $k$-SSP is easy.
\begin{prop}[\cite{LW08}]\label{ssp}
If $D=G=\f{q}$, then for any $g\in \f{q}$ and for any $1\le k\le q-1$ if $2\nmid q$ or for any $3\le k\le q-3$ if $2\mid q$, the $k$-SSP over $D$ always has solutions.
\end{prop}

The above proposition is applied to decide deep holes of $RS(q, k)$.
\begin{prop}[\cite{LW08}]
For any $2\le k\le q-4$, the vectors defined by polynomials of degree $k+1$ are not deep holes of $RS(q, k)$.
\end{prop}

For degree $k+2$ and odd $q$, a complete classification of deep holes for $RS(q,k)$ is obtained by Zhang, Fu and Liao~\cite{ZFL13}, and
detailed error distances are computed by Li and Zhu \cite{LZ15}.
\begin{prop}[\cite{ZFL13,LZ15}]\label{deepholek+2}
For any $2\le k\le q-3$, the vectors defined by polynomials of degree $k+2$ are not deep holes of $RS(q, k)$.
\end{prop}

For general degrees, the authors in~\cite{CM07} got the first result by reducing the conjecture to the
existence of rational points on a hypersurface over $\f{q}$.
Following a similar approach of Cheng-Wan~\cite{CW07}, Li and Wan~\cite{WL08} improved
the result in~\cite{CM07} with Weil's character sum estimate.
Later, Cafure et.al.~\cite{CMP11} improved the result in~\cite{WL08} a little bit by using tools of algebraic geometry.
Liao~\cite{Liao11} gave a tighter estimation of error distance, which was improved by Zhu and Wan~\cite{ZW12}.

A major progress in proving Conjecture~\ref{deepholeconj} is the recent result:
\begin{prop}[\cite{ZCL16}\cite{Kaipa17}]\label{CLZ15}
If $k+1\le p$ or $k\geq (q-1)/2$, Conjecture~\ref{deepholeconj} is true. In particular, the conjecture is true if $q$ is a prime.
\end{prop}

With the preparation above, we now give our new deep holes for some generalized RS codes, which can be viewed as a generalization of
the new deep holes found for even characteristic case in~\cite{ZFL13}.
We define
\begin{definition}\label{def:sumfreeset}
  We call a subset $D\subset\f{q}$ an $r$-zero-sum-free set in the finite field $\f{q}$ if for any $r$-subset $S$ of $D$,
  \[
  \sum_{s\in S}s\neq 0.
  \]
\end{definition}

Zero-sum-free sets are extensively studied in combinatorics and number theory (see~\cite{Erd65,TV16}).

\begin{exm}
	Let $p$ be a prime number and $r$ be an integer such that $1<r<p$. Then the set
	 \[
	  \left\{0,1,2,\cdots,\lfloor \frac{p}{r}\rfloor +r-1\right\}
	\]
	is an $r$-zero-sum-free set in the prime field $\f{p}$.
\end{exm}

\begin{exm}
Let $q$ be an odd prime power and let the evaluation set $D\subset \f{q}$ be such that
\[
   s\in D,\, s\neq 0\Rightarrow -s\notin D.
\]
That is, $D\cap \{-D\} \subseteq \{ 0\}$. Then $D$ is a $2$-zero-sum-free set.

\end{exm}

\begin{exm}
Let $q$ be a power of $2$ and let the evaluation set $D\subset \f{q}$ be arbitrary.
Then $D$ is a $2$-zero-sum-free set.
\end{exm}

Using the relationship between subset sum problem and deep holes defined by polynomials of degree $k+1$, we obtain the following new deep holes for generalized RS codes.
\begin{thm}\label{Newdeephole17a}
	Suppose the evaluation set $D\subset \f{q}$ is $r$-zero-sum-free ($r\geq 2$).
	Let $k=|D|-r-1$. For any polynomial $f_{\leq k-1}(x)$ of degree $\leq k-1$, the polynomial $x^{k+1}-(\sum_{s\in D}s) x^k+f_{\leq k-1}(x)$ defines a new class of deep holes of $RS(D,k)$ besides deep holes defined by polynomials of
	degree $k$ and polynomials in Proposition~\ref{newdeephole12}.
\end{thm}

 \begin{proof}
 First, we show that the word $x^{k+1}-(\sum_{s\in D}s) x^k+f_{\leq k-1}(x)$ really defines a deep hole of $RS(D,k)$. By lemma~\ref{SSPDeephole}, the word $x^{k+1}-(\sum_{s\in D}s) x^k+f_{\leq k-1}(x)$ does not define a deep hole of $RS(D,k)$ if and only if there is a $k+1$-subset of $D$ whose sum is $\sum_{s\in D}s$. The latter is equivalent to that the $r$ leftover elements in $D$ have sum $0$ which is impossible by the assumption on $D$. So $x^{k+1}-(\sum_{s\in D}s) x^k+f_{\leq k-1}(x)$ defines a deep hole of $RS(D,k)$.

 Second, we show the deep holes are new.
\begin{itemize}
  \item (Differ from polynomials of degree $k$.) If there exist some $a\in\f{q}^*$ and some polynomial $g_{\leq k-1}(x)$ of degree $\leq k-1$ such that
  $$x^{k+1}-(\sum_{s\in D}s) x^k+f_{\leq k-1}(x)=ax^k+g_{\leq k-1}(x)$$
   on $D$, then the polynomial of degree $k+1$
  $$x^{k+1}-(\sum_{s\in D}s) x^k+f_{\leq k-1}(x)-ax^k-g_{\leq k-1}(x)$$
  has at least $|D|=k+r+1$ distinct zeros, which implies it must be $0$. This contradicts to non-zero coefficients of the term $x^{k+1}$.
  \item (Differ from polynomials in Proposition~\ref{newdeephole12}.) If there exist some $a\in\f{q}^*$, $\delta\notin D$ and some polynomial $g_{\leq k-1}(x)$ of degree $\leq k-1$ such that
  $$x^{k+1}-(\sum_{s\in D}s) x^k+f_{\leq k-1}(x)=\frac{a}{x-\delta}+g_{\leq k-1}(x)$$
   on $D$. As $x-\delta$ is non-zero everywhere on $D$, we multiply $x-\delta$ on the equality to get
   $$(x^{k+1}-(\sum_{s\in D}s) x^k+f_{\leq k-1}(x))(x-\delta)=a+(x-\delta)g_{\leq k-1}(x)$$
   on $D$. So the polynomial of degree $k+2$
   $$(x^{k+1}-(\sum_{s\in D}s) x^k+f_{\leq k-1}(x))(x-\delta)-a-(x-\delta)g_{\leq k-1}(x)$$
   has at least $k+r+1$ distinct zeros, which implies that it must be $0$. By the same reason above, we obtain a contradiction.

\end{itemize}

So the deep holes are new.

 \end{proof}

\section{Covering Radii and Deep Holes for PRS codes}\label{PRSCode}
In this section, we study the covering radii and deep holes for PRS codes.
Recall that the PRS code $PRS(q+1,k)$ has one generator matrix of the form
\begin{equation*}
  \left(
    \begin{array}{ccccc}
      1 & 1 & \cdots & 1 & 0 \\
      \alpha_1 & \alpha_2 & \cdots & \alpha_q & 0 \\
      \vdots & \vdots & \ddots & \vdots & \vdots \\
      \alpha_1^{k-2} & \alpha_2^{k-2} & \cdots & \alpha_q^{k-2} & 0 \\
      \alpha_1^{k-1} & \alpha_2^{k-1} & \cdots & \alpha_q^{k-1} & 1 \\
    \end{array}
  \right).
\end{equation*}
It is easy to check that the PRS code has minimum distance $q+2-k$ and thus it is also an MDS code.

For the case $k=1$, the PRS code $PRS(q+1,1)$ is nothing but the repeating code generated by $(1,1,\cdots,1)$.
In this case, one can easily show that the covering radius is $d-2=q-1$ and the deep holes are
permutations of the $(q+1)$-element multi-set $\f{q}\cup \{\alpha\}$, where $\alpha\in\f{q}$ is arbitrary.

For the case $k=q-1$, the proof of Theorem~\ref{thm1} in~\cite{ZW16} shows that the covering radius of $PRS(q+1,k)$ is $d-2=1$ and the deep holes are $(a,\cdots,a,0,v)+PRS(q+1,k)$, where $a\in\f{q}^*$, and $v\in \f{q}$ are arbitrary.

For the case $k=q$, one can show that the covering radius of $PRS(q+1,k)$ is $d-1=1$ and the deep holes are $w+PRS(q+1,k)$ for all $w\in\f{q}^{q+1}$ of weight $1$.

With the boundary cases removed, we can then assume that  $2\le k\le q-2$.

\subsection{Covering radii of PRS codes}
Although the covering radius of RS codes is always $d-1$, it seems a little
surprising that the covering radius of PRS codes is unknown in general. The above boundary cases suggest that the covering radius
of $PRS(q+1,k)$ can have two possibilities $q-k, q-k+1$, which is indeed the case.
Another example in~\cite{CKMS85} is the PRS code $C$ over $\f{5}$ with generator matrix
 \begin{equation*}
\left(
  \begin{array}{cccccc}
    1 & 1 & 1 & 1 & 1 & 0 \\
    1 & 2 & 3 & 4 & 0 & 0 \\
    1 & 2^2 & 3^2 & 4^2 & 0 & 0 \\
    1 & 2^3 & 3^3 & 4^3 & 0 & 1 \\
  \end{array}
\right).
\end{equation*}
The code $C$ has minimum distance $3$ and covering radius $1$.
 This example suggests $PRS(q+1,k)$ may have covering radius $q-k$, two smaller than the minimum distance $q+2-k$.
 This leads to Conjecture~\ref{cover}.
In~\cite{Arne94}, D\"{u}r proved

\begin{prop}[\cite{Arne94}]\label{old}
If $q$ is an odd prime power and $2\leq k<\frac{\sqrt{q}}{4}+\frac{39}{16}$ or $6\sqrt{q\ln q}-2\leq k\leq q-2$; or if $q$ is a power of $2$ and $5\leq k<\frac{\sqrt{q}}{2}-\frac{11}{4}$ or $6\sqrt{q\ln q}-2\leq k\leq q-3$, then the covering radius of $PRS(q+1,k)$ is
\[
   \rho(PRS(q+1,k))=q-k.
\]
\end{prop}

In~\cite{ZW16} we studied the relationship between the covering radius Conjecture~\ref{cover} for PRS codes and the deep hole Conjecture~\ref{deepholeconj}.

\begin{thm}[\cite{ZW16}]\label{thm1}
Let $q$ be an odd prime power. Assume that Conjecture~\ref{deepholeconj} is true. Then the covering radius of $PRS(q+1,k)$ is $q-k$.
\end{thm}

Note that the above theorem holds for even $q$: if Conjecture~\ref{deepholeconj} is true, then the covering radius of $PRS(q+1,k)$ is $q-k$ for all $3\leq k\leq q-3$. The proof is the same as the proof in~\cite{ZW16}.

Using recent results on Conjecture~\ref{deepholeconj}~\cite{Kaipa17}, we can deduce

\begin{cor}[\cite{ZW16}]
Let $q$ be a power of the prime $p$. Assume either $3\le k+1\le p$ or $k\geq  (q-1)/2$ (and $k\neq q-2$ if $p=2$). Then the covering radius of $PRS(q+1,k)$ is $q-k$.
\end{cor}

It is natural to ask if the converse holds. We only have partial result.

\begin{thm}\label{dhk-ndh1}
Let $q$ be a prime power and let $k$ be an integer such that $2\leq k\leq q-2$ if $2\nmid q$ or $3\leq k\leq q-3$ if $2\mid q$. If Conjecture~\ref{cover} holds, deep holes of $RS(q,k)$ are not deep hole
of $RS(q,k-1)$.
\end{thm}
\begin{proof}
Consider the word $(u_f,v)$, where $u_f$ is a deep hole of $RS(q,k)$ and $v\in \f{q}$
 is arbitrary. By the assumption $\rho(PRS(q+1,k))=q-k$, we get
 \[
   d((u_f,v),PRS(q+1,k))\leq q-k.
 \]
That is, there is a codeword $(u_h,c_{k-1}(h))\in PRS(q+1,k)$ such that
 \[
   d((u_f,v),(u_h,c_{k-1}(h)))\leq q-k.
 \]
 On the other hand, we have
 \begin{align*}
    d((u_f,v),(u_h,c_{k-1}(h)))&=d(u_f,u_h)+d(v,c_{k-1}(h))\\
    &\geq q-k+d(v,c_{k-1}(h)),
 \end{align*}
 since $d(u_f,RS(q,k))=q-k$ and $u_h\in RS(q,k)$. It forces
 \[
   d(u_f,u_h)=q-k,\qquad\mbox{and}\qquad c_{k-1}(h)=v.
 \]
 So
 \[
    d(u_{f-vx^{k-1}},u_{h-vx^{k-1}})=d(u_f,u_h)=q-k.
 \]
 As $h-vx^{k-1}$ is a polynomial of degree $\leq k-2$, it follows
 \[
    d(u_{f-vx^{k-1}}, RS(q,k-1))\leq q-k.
 \]
Because $\rho(RS(q,k-1))=q-k+1$, $u_{f-vx^{k-1}}$ is not a deep hole of $RS(q,k-1)$.
In particular, taking $v=0$ we prove the lemma.

\end{proof}

\subsection{Deep holes of PRS codes}
A further question is to find all the deep holes of the PRS codes. On the opposite side, the following result shows that a large class of vectors are not deep holes.

\begin{thm}[\cite{ZW16}]\label{deepholes}
	If $\deg(f)\ge k+1$, denote $s=\deg(f)-k+1$. There are positive constants $c_1$ and $c_2$ such that if
	$$s<c_1\sqrt{q},\,\left(\frac{s}{2}+2\right)\log_2(q)<k<c_2q,$$
	then for any $v\in \f{q}$, $(u_f,v)$ is not a deep hole of $PRS(q+1,k)$.
\end{thm}

On the positive side, we have the following construction of deep holes of PRS codes.

\begin{thm}[Construction from irreducible quadratic polynomials]\label{deephole:irr2}
	Let $q$ be a prime power and let $k$ be an integer such that $2\leq k\leq q-2$ if $2\nmid q$ or $3\leq k\leq q-3$ if $2\mid q$. Suppose $\rho(PRS(q+1,k))=q-k$. Let $p(x)=x^2+\alpha x+\beta\in \f{q}[x]$ be a monic irreducible polynomial over $\f{q}$. For any $a,b\in \f{q}$ which are not zero simultaneously, the rational function $\frac{a+bx}{p(x)}$ defines the following coset of deep holes of the PRS code $PRS(q+1,k)$:
	\[
	  \left(\frac{a+b\alpha_1}{p(\alpha_1)},\frac{a+b\alpha_2}{p(\alpha_2)},\cdots,\frac{a+b\alpha_q}{p(\alpha_q)},0\right)+C.
	\]
\end{thm}
\begin{proof}
	We will define a new MDS code $C'$ with parameters $[q+1,k+2]$ which contains $C=PRS(q+1,k)$ as a subcode. So all the words in $C'\setminus C$ are deep holes of $C$.
	
	The code $C'$ is generated by the generator matrix
	\begin{equation*}
	\left(
	\begin{array}{ccccc}
	\frac{1}{p(\alpha_1)} & \frac{1}{p(\alpha_2)} & \cdots & \frac{1}{p(\alpha_q)} & 0 \\
	\frac{\alpha_1}{p(\alpha_1)} & \frac{\alpha_2}{p(\alpha_2)} & \cdots & \frac{\alpha_q}{p(\alpha_q)} & 0 \\
	\vdots & \vdots & \ddots & \vdots & \vdots \\
	\frac{\alpha_1^{k}}{p(\alpha_1)} & \frac{\alpha_2^{k}}{p(\alpha_2)} & \cdots & \frac{\alpha_q^{k}}{p(\alpha_q)} & 0 \\
	\frac{\alpha_1^{k+1}}{p(\alpha_1)} & \frac{\alpha_2^{k+1}}{p(\alpha_2)} & \cdots & \frac{\alpha_q^{k+1}}{p(\alpha_q)} & 1 \\
	\end{array}
	\right).
	\end{equation*}
	From the definition, it is easy to see that $C'$ is equivalent to the PRS code $[q+1,k+2]$ code with the scale vector
	$$\left(\frac{1}{p(\alpha_1)}, \frac{1}{p(\alpha_2)}, \cdots, \frac{1}{p(\alpha_q)}, 1\right).$$
	So $C'$ is an MDS code with parameters $[q+1,k+2]$.
	
	The PRS code $C=PRS(q+1,k)$ is a subcode of the new code $C'$, as
	\begin{align*}
  \left(
\begin{array}{ccccc}
1 & 1 & \cdots & 1 & 0 \\
\alpha_1 & \alpha_2 & \cdots & \alpha_q & 0 \\
\vdots & \vdots & \ddots & \vdots & \vdots \\
\alpha_1^{k-2} & \alpha_2^{k-2} & \cdots & \alpha_q^{k-2} & 0 \\
\alpha_1^{k-1} & \alpha_2^{k-1} & \cdots & \alpha_q^{k-1} & 1 \\
\end{array}
\right)=
  \left(
\begin{array}{cccccccc}
\beta & \alpha & 1 & 0& \cdots & 0 & 0 & 0 \\
0 & \beta & \alpha & 1& \cdots &0 &0 & 0 \\
\vdots & \vdots &\vdots &\vdots & \ddots & \vdots& \vdots & \vdots \\
0 & 0 & 0&0& \cdots & \beta&\alpha & 1 \\
\end{array}
\right)\\\times\left(
\begin{array}{ccccc}
\frac{1}{p(\alpha_1)} & \frac{1}{p(\alpha_2)} & \cdots & \frac{1}{p(\alpha_q)} & 0 \\
\frac{\alpha_1}{p(\alpha_1)} & \frac{\alpha_2}{p(\alpha_2)} & \cdots & \frac{\alpha_q}{p(\alpha_q)} & 0 \\
\vdots & \vdots & \ddots & \vdots & \vdots \\
\frac{\alpha_1^{k}}{p(\alpha_1)} & \frac{\alpha_2^{k}}{p(\alpha_2)} & \cdots & \frac{\alpha_q^{k}}{p(\alpha_q)} & 0 \\
\frac{\alpha_1^{k+1}}{p(\alpha_1)} & \frac{\alpha_2^{k+1}}{p(\alpha_2)} & \cdots & \frac{\alpha_q^{k+1}}{p(\alpha_q)} & 1 \\
\end{array}
\right).
\end{align*}	

For any non-zero vector $c'\in C'\setminus C$, we have
\[
 q-k= d(C')\leq d(c',C) \leq \rho(C)=q-k,
\]
so
\[
  d(c',C)= \rho(C)=q-k.
\]
Note that any non-zero vector $c'\in C'\setminus C$ has the form
$$c'=\left(\frac{a+b\alpha_1}{p(\alpha_1)},\frac{a+b\alpha_2}{p(\alpha_2)},\cdots,\frac{a+b\alpha_q}{p(\alpha_q)},0\right)+c$$
for some $a,b\in \f{q}$ which are not zero simultaneously and some $c\in C$.

So for any $a,b\in \f{q}$ which are not zero simultaneously, the vectors
\[
\left(\frac{a+b\alpha_1}{p(\alpha_1)},\frac{a+b\alpha_2}{p(\alpha_2)},\cdots,\frac{a+b\alpha_q}{p(\alpha_q)},0\right)+C
\]
are deep holes of the PRS code $PRS(q+1,k)$.
\end{proof}

We study the relationship between deep holes constructed from monic irreducible polynomials of degree $2$ in the above theorem and the deep holes defined by polynomials of degree $k$ in Theorem~\ref{deephole:degreek}.
\begin{thm}
Let $q$ be a prime power and let $k$ be an integer such that $2\leq k\leq q-2$ if $2\nmid q$ or $3\leq k\leq q-3$ if $2\mid q$. Suppose $\rho(PRS(q+1,k))=q-k$. If $k\leq q-3$,
	the deep holes $\{(u_f,v)\,|\,\deg(f)=k,\,v\in\f{q}\}$ are different from the deep holes in Theorem~\ref{deephole:irr2}. If $q$ is odd and $k=q-2$, the deep holes $\{(u_f,v)\,|\,\deg(f)=k,\,v\in\f{q}\}$ are contained in the deep holes in Theorem~\ref{deephole:irr2}.
\end{thm}
\begin{proof}
	Wlog, we may assume $f(x)=x^k-wx^{k-1}$, as $(u_{x^i},0)\in PRS(q+1,k)$ for all $0\leq i\leq k-2$ and any non-zero scale multiple of a deep hole is still a deep hole. Let
	 \[
	  g(x)=f(x)-vx^{k-1}=x^k-(w+v)x^{k-1}.
	\]
	Then we have
	\[
	d((u_g,0),PRS(q+1,k))=d((u_f,v),PRS(q+1,k)).
	\]
	
	For the case $k\leq q-3$, suppose there are some monic irreducible polynomial $p(x)\in\f{q}[x]$ of degree $2$, some non-zero polynomial $a+bx\in\f{q}[x]$ and some polynomial $h(x)\in\f{q}[x]$ of degree $\leq k-2$ such that
	\[
	   g(x)=\frac{a+bx}{p(x)}+h(x)\qquad \forall x\in \f{q}.
	\]
	Then the polynomial of degree $k+2\leq q-1$
	\[
	  (g(x)-h(x))p(x)-(a+bx)
	\]
	has $q$ different zeros (as it is zero on $\f{q}$), which is impossible.
	So in this case, the deep holes $\{(u_f,v)\,|\,\deg(f)=k,\,v\in\f{q}\}$ are different from the deep holes in Theorem~\ref{deephole:irr2}.
	
	For the case $k=q-2$, let $p(x)=x^2+\alpha x+\beta \in \f{q}[x]$ be any monic irreducible polynomial over $\f{q}$. Let
	\[
	a+bx\equiv -(x^q-x)\mod p(x).
	\]
	Since $x^q-x$ is coprime to $p(x)$, the polynomial $a+bx$
	is non-zero. There exists some polynomial $h(x)\in\f{q}[x]$ of degree $\leq k-2$ such that
	\[
	(x^k-\alpha x^{k-1}-h(x))p(x)={a+bx}+x^q-x.
	\]
	Hence by choosing a monic irreducible polynomial\footnote{If $q$ is odd (or $q$ is even and $w+v=0$), one can always use the discriminant to determine if $p(x)$ is irreducible. That is, if $(w+v)^2-4\beta$ is a non-square in $\f{q}$, then $p(x)$ is irreducible over $\f{q}$. So in this case, there are $\frac{q-1}{2}$ many of $\beta\in\f{q}$ making $(w+v)^2-4\beta$ a non-square for given $w+v$. However the discriminant method is invalid for even $q$ and $w+v\neq 0$. We show the existence of $\beta$ by contradiction. Suppose for all $\beta\in\f{q}$, the quadratic polynomial $p(x)=x^2+(w+v) x+\beta$ is reducible over $\f{q}$. So $p(x)$ has two roots in $\f{q}$, saying $\gamma_1,\gamma_2$. Then $\gamma_1\neq \gamma_2$, as $\gamma_1+\gamma_2=w+v\neq 0$. On the other hand, for any $\gamma\in\f{q}$ there is at most one $\beta\in\f{q}$ such that $p(\gamma)=0$. So when $\beta$ runs over $\f{q}$ there are totally $2q$ pairwise distinct roots in $\f{q}$ which is impossible.}
	$$p(x)=x^2+(w+v) x+\beta \in \f{q}[x],$$ there is some polynomial $h(x)\in\f{q}[x]$ of degree $\leq k-2$ such that
	\[
	(g(x)-h(x))p(x)={a+bx}+x^q-x,
	\]
	which is also equivalent to
	\[
	g(x)=\frac{a+bx}{p(x)}+h(x)\qquad \forall x\in \f{q}.
	\]
	So in this case, the deep holes $\{(u_f,v)\,|\,\deg(f)=k,\,v\in\f{q}\}$ are contained in the deep holes in Theorem~\ref{deephole:irr2}.

\end{proof}

The following lemma shows that deep holes in Theorem~\ref{deephole:irr2} are disjoint for different monic irreducible polynomials of degree $2$ in the case $\leq q-3$. And the case $k=q-2$ will be studied in detail in Section~\ref{section:k=q-2}.

\begin{lem}\label{disjoint:degree2}
	Suppose $2\leq k\leq q-3$. For any non-zero polynomials $a_1+b_1x, a_2+b_2x\in \f{q}[x]$ and any two different monic irreducible polynomials $p_1(x), p_2(x)\in \f{q}[x]$ of degree $2$, there is no polynomial $f(x)\in \f{q}[x]$ of degree $\leq k-2$ such that
	\[
	  \frac{a_1+b_1x}{p_1(x)}=\frac{a_2+b_2x}{p_2(x)}+f(x)\qquad\forall x\in \f{q}.
	\]
\end{lem}
\begin{proof}
	Suppose there is a polynomial $f(x)\in \f{q}[x]$ of degree $\leq k-2$ such that
	\[
	\frac{a_1+b_1x}{p_1(x)}=\frac{a_2+b_2x}{p_2(x)}+f(x)\qquad\forall x\in \f{q}.
	\]
	Denote the polynomial
	\[
	   g(x)=(a_1+b_1x)p_2(x)-(a_2+b_2x)p_1(x)+f(x)p_1(x)p_2(x).
	\]
	Then we have $$\deg(g(x))=\deg(f(x))+4\leq k+2\leq q-1$$ and
	\[
	x^q-x\,|\,g(x).
	\]
	So $g(x)$ must be the zero polynomial. It enforces $f(x)=0$ and
	\[
	 (a_1+b_1x)p_2(x)=(a_2+b_2x)p_1(x)
	\]
	which is impossible under our assumption.
\end{proof}

   \subsection{The number of deep holes of PRS codes}
	A classification result of Kaipa~\cite{Kaipa17} about deep holes of $PRS(q+1,k)$, if $\rho(PRS(q+1,k))=q-k$, could allow us to count the number of cosets of deep holes. Consider the syndrome of the received word $u\in \f{q}^{q+1}$ defined as follows:
	\[
	Hu^T= \left(
	\begin{array}{ccccc}
	1 & 1 & \cdots & 1 & 0 \\
	\alpha_1 & \alpha_2 & \cdots & \alpha_q & 0 \\
	\vdots & \vdots & \ddots & \vdots & \vdots \\
	\alpha_1^{q-k-1} & \alpha_2^{q-k-1} & \cdots & \alpha_q^{q-k-1} & 0 \\
	\alpha_1^{q-k} & \alpha_2^{q-k} & \cdots & \alpha_q^{q-k} & 1 \\
	\end{array}
	\right)u^T.
	\]
	If $\rho(PRS(q+1,k))=q-k$, the coset $u+PRS(q+1,k)$ is a coset of deep holes of $PRS(q+1,k)$ if and only if the syndrome $Hu^T$ is not in the span of any $q-k-1$ vectors in the normal rational curve
	\[
	\mbox{NRC}=\{(1,x,\cdots,x^{q-k})\,|\,x\in\f{q}\}\cup \{(0,\cdots,0,1)\}.
	\]
	
	\begin{thm}[\cite{Kaipa17}]\label{numberofdeepholes:k=q-2}
		Let $q$ be an odd prime power and $k=q-2$. There are $(q-1)q^2$ cosets of deep holes of $PRS(q+1,k)$.
	\end{thm}
	
	\begin{thm}\label{numberofdeepholes:k=q-3}
		Let $q$ be a prime power and $k=q-3$. There are $(q-1)\left(\frac{q^3}{2}+q^2+\frac{q}{2}\right)$ cosets of deep holes of $PRS(q+1,k)$.
	\end{thm}
	\begin{proof}
		Since any $4$ vectors in the $\mbox{NRC}$
			\[
		\left(
		\begin{array}{ccccc}
		1 & 1 & \cdots & 1 & 0 \\
		\alpha_1 & \alpha_2 & \cdots & \alpha_q & 0 \\
		\alpha_1^{2} & \alpha_2^{2} & \cdots & \alpha_q^{2} & 0 \\
		\alpha_1^{3} & \alpha_2^{3} & \cdots & \alpha_q^{3} & 1 \\
		\end{array}
		\right)
		\] are linearly independent over $\f{q}$, the union of all the span of $2$ vectors in $\mbox{NRC}$ has cardinality
		\[
		\binom{q+1}{2}(q-1)^2+(q+1)(q-1)+1.
		\]
		So the number of cosets of deep holes is
		\[
		q^{4}-\left(\binom{q+1}{2}(q-1)^2+(q+1)(q-1)+1\right)=(q-1)\left(\frac{q^3}{2}+q^2+\frac{q}{2}\right).
		\]
	\end{proof}

	\section{Complete deep holes of $PRS(q+1,q-2)$}\label{section:k=q-2}
	In this section, we focus on the case $q$ is odd and $k=q-2$. In this case, the covering radius of $PRS(q+1,q-2)$ is deterministic, saying $\rho(PRS(q+1,q-2))=2$. In~\cite{Kaipa17}, Kaipa gave a classification of the syndromes of deep holes of $PRS(q+1,q-2)$ using the $PGL(2,q)$ group action on the projective space $PG(2,q)$, but without explicit construction of the deep holes.
	In this section, we will show that deep holes in Theorem~\ref{deephole:irr2} are all the deep holes of $PRS(q+1,q-2)$. So all the deep holes of $PRS(q+1,q-2)$ are constructed explicitly.
	
	We first give some facts about the deep holes found in Theorem~\ref{deephole:irr2}. For any monic irreducible quadratic polynomial $p(x)\in\f{q}[x]$ we denote
	\[
	DH(p(x))=\left\{\frac{a+bx}{p(x)}\,|\,a,b\in\f{q}\,\mbox{are not zero simultaneously}\right\}.
	\]
	As there is one-to-one correspondence between rational functions in $DH(p(x))$ and the cosets of deep holes of the form
	\[
	\left(\frac{a+b\alpha_1}{p(\alpha_1)},\frac{a+b\alpha_2}{p(\alpha_2)},\cdots,\frac{a+b\alpha_q}{p(\alpha_q)},0\right)+C,
	\] we also denote by $DH(p(x))$ the cosets of deep holes of the above form defined by $p(x)$.
	
	\textbf{Fact 1.} We have
	\[
	  |DH(p(x))|=q^2-1,
	\]
	and
	\[
	  |DH(p_1(x))\cap DH(p_2(x))|=q-1
	\]
	for any two distinct $p_1(x),p_2(x)$.
	
	\begin{proof}
		It is obvious that
			\[
		|DH(p(x))|=q^2-1.
		\]
		The same argument in the proof of Lemma~\ref{disjoint:degree2} shows that the only possibility for
     	\[
        \frac{a_1+b_1x}{p_1(x)}=\frac{a_2+b_2x}{p_2(x)}+f(x)\qquad\forall x\in \f{q}.
        \]
		 is that $\deg(f(x))=k-2=q-4$ and
		\[
		g(x)=a(x^q-x)
		\]
		for some $a\in\f{q}^*$. So
		\[
		p_1(x)p_2(x)\,|\,a(x^q-x)-(a_1+b_1x)p_2(x)+(a_2+b_2x)p_1(x).
		\]
		Since $ p_1(x)$, and $p_2(x)$ are two distinct monic irreducible polynomial over $\f{q}$, the condition above is equivalent to
		\[\begin{cases}
		p_1(x)\,|\,a(x^q-x)-(a_1+b_1x)p_2(x),\\
		p_2(x)\,|\,a(x^q-x)+(a_2+b_2x)p_1(x),
		\end{cases}\]
		which has solution as follows
		\[\begin{cases}
		a_1+b_1x\equiv \frac{a(x^q-x)}{p_2(x)}\mod p_1(x),\\
		a_2+b_2x\equiv -\frac{a(x^q-x)}{p_1(x)}\mod p_2(x).
		\end{cases}\]
		So when $a$ runs over $\f{q}^*$, it gives the $q-1$ elements of $DH(p_1(x))\cap DH(p_2(x))$. And hence,
		 	\[
		 |DH(p_1(x))\cap DH(p_2(x))|=q-1.
		 \]
	\end{proof}

	\textbf{Fact 2.} For any three monic irreducible polynomials $p_1(x),p_2(x), p_3(x)$ of degree $2$, the condition
	\[
	  DH(p_1(x))\cap DH(p_2(x))\cap DH(p_3(x))\neq \emptyset
	\]
	if and only if $p_1(x),p_2(x), p_3(x)$ are linearly dependent, i.e.,
	there is unique $c\in\f{q}\setminus\{0,1\}$ such that
	\[
	p_1(x)=cp_2(x)+(1-c)p_3(x).
	\]

	\begin{proof}
		Suppose
		\[
		\frac{a_1+b_1x}{p_1(x)}\in DH(p_1(x))\cap DH(p_2(x))\cap DH(p_3(x)).
		\]
		Then there are $c_1,c_2\in\f{q}^*$ such that
		\[
		a_1+b_1x\equiv \frac{c_1(x^q-x)}{p_2(x)}\mod p_1(x),
		\]
		and
		\[
		a_1+b_1x\equiv \frac{c_2(x^q-x)}{p_3(x)}\mod p_1(x).
		\]
        So
		\[
		\frac{c_2(x^q-x)}{p_3(x)}\equiv \frac{c_1(x^q-x)}{p_2(x)}\mod p_1(x)
		\]
		or
		\[
		p_1(x)\,|\,(c_2p_2(x)-c_1p_3(x))(x^q-x).
		\]
		Since $x^q-x$ is coprime to $p_1(x)$, we have
		\[
		p_1(x)\,|\,c_2p_2(x)-c_1p_3(x).
		\]
		So $c_1\neq c_2$ and
		\[
		(c_2-c_1)p_1(x)=c_2p_2(x)-c_1p_3(x).
		\]
		Let $c=\frac{c_2}{c_2-c_1}$. Then
		$$c\in\f{q}\setminus\{0,1\},$$
		and
		\[
		p_1(x)=cp_2(x)+(1-c)p_3(x).
		\]
		If there is another $c'\in\f{q}\setminus\{0,1\}$ such that	
			\[
		p_1(x)=c'p_2(x)+(1-c')p_3(x).
		\]
		Then
		\[
		cp_2(x)+(1-c)p_3(x)=c'p_2(x)+(1-c')p_3(x)
		\]
		which is impossible unless $c'=c$. So $c$ is unique.
		
		Conversely, if there is $c\in\f{q}\setminus\{0,1\}$ such that
		\[
		p_1(x)=cp_2(x)+(1-c)p_3(x),
		\]
		then
		\[
		DH(p_1(x))\cap DH(p_2(x))\cap DH(p_3(x))\neq \emptyset.
		\]
		Let
		\[
		a_1+b_1x\equiv\frac{x^q-x}{p_2(x)}\mod p_1(x).
		\]
		Then
		\[
		\frac{a_1+b_1x}{p_1(x)}\in DH(p_1(x))\cap DH(p_2(x))
		\]
		By the assumption, we have
		\[
		\frac{x^q-x}{p_2(x)}\equiv \frac{c}{c-1}\frac{x^q-x}{p_3(x)}\mod p_1(x).
		\]
		So
		\[
		a_1+b_1x\equiv\frac{c}{c-1}\frac{x^q-x}{p_3(x)}\mod p_1(x),
		\]
		which implies
		\[
		\frac{a_1+b_1x}{p_1(x)}\in DH(p_1(x))\cap DH(p_3(x))
		\]
		Hence,
		\[
		\frac{a_1+b_1x}{p_1(x)}\in DH(p_1(x))\cap DH(p_2(x))\cap DH(p_3(x)).
		\]
		
	\end{proof}
	
	\textbf{Fact 3.} For any pairwise distinct monic irreducible quadratic polynomials $p_1(x),p_2(x),\cdots,p_j(x)$ if
	\[
	|DH(p_1(x))\cap DH(p_2(x))\cap \cdots\cap DH(p_j(x))|\neq \emptyset,
	\]
	then
	\[
	|DH(p_1(x))\cap DH(p_2(x))\cap \cdots\cap DH(p_j(x))|=q-1.
	\]
	
   Based on the above facts, we construct a hypergraph concerning of the distribution of irreducible quadratic polynomials, which makes the picture of deep holes more clear. Since $DH(p(x))$ is closed under any non-zero scaling, we consider two polynomials
   \[
     \frac{a+bx}{p(x)}\qquad\mbox{and}\qquad c\frac{a+bx}{p(x)},\quad c\in\f{q}^*
   \]
	as the same vertex in the hypergraph $G=(V,E)$. So the vertices $V$ of the hypergraph are defined to be
	\[
	  \bigcup_{p(x)}DH(p(x))/\f{q}^*
	\]
	where $p(x)$ runs through all monic irreducible polynomials of degree $2$. Any edge $e\in E$ is defined to be
	\[
	  e=DH(p(x))/\f{q}^*
	\]
	for some monic irreducible polynomial $p(x)$ of degree $2$.
	By Fact~1, any edge has $q+1$ vertices and any two edges intersects with $1$ vertex.  Facts~2 tells that any vertex has degree about $\frac{q-1}{2}$, as there are about half $c\in\f{q}$ making $p_1(x)$ irreducible if given $p_2(x)$ and $p_3(x)$.

	\begin{lem}\label{graph:2degrees}
	Any vertex $v$ in the hypergraph $G$ has degree $\frac{q-1}{2}$ or $\frac{q+1}{2}$.
	\end{lem}
	
     \begin{proof}
     	For any fixed vertex $v=\frac{a_1+b_1x}{p_1(x)}\in DH(p_1(x))$ ($p_1(x)=x^2+\alpha_1 x+\beta_1$), we need to count the number of monic irreducible polynomials $p_2(x)=x^2+\alpha_2 x+\beta_2$ ($p_2(x)\neq p_1(x)$) such that
     	\begin{equation}\label{conditionofintersection}
     	a_1+b_1x\equiv \frac{c(x^q-x)}{p_2(x)}\mod p_1(x)\tag{*}
     	\end{equation}
     	for some $c\in\f{q}^*$.
     	
     	Because $a_1+b_1x$ is coprime to $p_1(x)$, let $f(x)$ be the inverse of $a_1+b_1x$ modulo $p_1(x)$, i.e.,
     	\[
     	 (a_1+b_1x)f(x)\equiv 1 \mod p_1(x).
     	\]
     	Since $f(x)$ and $p_2(x)$ are coprime to $p_1(x)$, we multiply $f(x)p_2(x)$ on both sides of the equation~(\ref{conditionofintersection}) and obtain
     	\[
     	p_2(x)\equiv c(x^q-x)f(x)\mod p_1(x).
     	\]
     	Denote
     	\[
     	  e_1x+e_2\equiv c(x^q-x)f(x)\mod p_1(x).
       	\]
     	Then $e_1, e_2$ are non-zero simultaneously. The condition~(\ref{conditionofintersection}) is equivalent to
     	\[
          p_2(x)\equiv c(e_1x+e_2)\mod p_1(x).
     	\]
        So the existence of $p_2(x)$ and $c\in\f{q}^*$ is equivalent to that the discriminant
        \begin{align*}
         \Delta&=(\alpha_1+e_1c)^2-4(\beta_1+e_2c)\\
               &=e_1^2c^2+(2\alpha_1e_1-4e_2)c+(\alpha_1^2-4\beta_1)\\
        \end{align*}
     	is a non-square for some $c\in\f{q}^*$. Hence the degree of the vertex $v$ is
     	\[
     	\deg(v)=|\{c\in\f{q}^*\,|\,\Delta \mbox{ is a non-square}\}|+1.
     	\]
     	Since $c=0$ makes $\Delta$ a non-square, we have
     	\[
     	\deg(v)=|\{c\in\f{q}\,|\,\Delta \mbox{ is a non-square}\}|.
     	\]
     	
     	\begin{itemize}
     		\item If $e_1=0$, then $e_2\neq 0$ and there are $\frac{q-1}{2}$ many of $c\in\f{q}$ such that $\Delta$ is a non-square. So in this case,
     		\[
     		\deg(v)=\frac{q-1}{2}
     		\]
     		\item If $e_1\neq 0$, then the discriminant
     		\begin{align*}
     		\Delta&=e_1^2c+(2\alpha_1e_1-4e_2)c+(\alpha_1^2-4\beta_1)\\
     		&=\left(e_1c+\frac{\alpha_1e_1-2e_2}{e_1}\right)^2+ (\alpha_1^2-4\beta_1)-\left(\frac{\alpha_1e_1-2e_2}{e_1}\right)^2
     		\end{align*}
     		Denote
     		\[
     		\delta=(\alpha_1^2-4\beta_1)-\left(\frac{\alpha_1e_1-2e_2}{e_1}\right)^2.
     		\]
     		Then after a linear transformation of $\f{q}$,
     		\[
     		\deg(v)=|\{c\in\f{q}\,|\,c^2-{\delta}\,\mbox{ is a non-square in}\, \f{q}\}|.
     		\]	
     		Since $\alpha_1^2-4\beta_1$ is a non-square, the constant term is non-zero, i.e.,
     		\[
     		\delta\neq 0.
     		\]
     		Next, we count the number of $c\in\f{q}$ that makes $c^2-{\delta}$ a square. Suppose
     		\[
     		c^2-{\delta}=y^2
     		\]
     		for some $y\in\f{q}$. Then
     		\[
     		\delta=(c-y)(c+y).
     		\]
     		It is easy to see that the above equation has solutions
     		\[
     		c=\frac{\delta/l+l}{2}, y=\frac{\delta/l-l}{2},\,\,\forall l\in\f{q}^*.
     		\]
     		So
     		\[
     		|\{c\in\f{q}\,|\,c^2-{\delta}\,\mbox{ is a square in}\, \f{q}\}|=\left|\left\{\frac{\delta/l+l}{2}\,|\,l\in\f{q}^*\right\}\right|.
     		\]
     		Hence,
     		\begin{align*}
     		\deg(v)&=|\{c\in\f{q}\,|\,c^2-{\delta}\,\mbox{ is a non-square in}\, \f{q}\}|\\
     		&=q-\left|\left\{\frac{\delta/l+l}{2}\,|\,l\in\f{q}^*\right\}\right|\\
     		&=\begin{cases}
     		\frac{q-1}{2}\,\qquad\mbox{if $\delta$ is a square;}\\
     		\frac{q+1}{2}\,\qquad\mbox{if $\delta$ is a non-square.}
     		\end{cases}
     		\end{align*}
     	\end{itemize}
     	
     \end{proof}

	 \begin{lem}\label{distributionofirreducible}
	    In any edge of the hypergraph $G$, half of the vertices have degree $\frac{q+1}{2}$, and the other half of the vertices have degree $\frac{q-1}{2}$.
	 \end{lem}
	
	\begin{proof}
	Let $e=DH(p(x))/\f{q}^*$ be an edge of $G$. Suppose there are $m$ vertices in $e$ of degree $\frac{q-1}{2}$ and $n$ vertices in $e$ of degree $\frac{q+1}{2}$. There are totally $\frac{q^2-q}{2}$ edges in $G$. By Fact~1 besides the edge $e$, any of the other $\frac{q^2-q}{2}-1$ edges intersects with $e$. So  
    \[
    m\left(\frac{q-1}{2}-1\right)+n\left(\frac{q+1}{2}-1\right)=\frac{q^2-q}{2}-1.
    \]	
   By Lemma~\ref{graph:2degrees}, we have
   \[ m+n=q+1.\]
		So
		\[
		 m=n=\frac{q+1}{2}.
		\]
	\end{proof}
	
		\begin{thm}\label{numberofvertices}
		The hypergraph $G$ has $|V|=q^2$ vertices.
	\end{thm}
	\begin{proof}
		As a multi-set, the union of all the edges has cardinality
		\[
		 |\{\mbox{monic irreducible polynomials of degree $2$}\}|(q+1)=\frac{q^2-q}{2}(q+1).
		\]
		By Lemma~\ref{distributionofirreducible}, half of the elements in the multi-set have multiplicity $\frac{q+1}{2}$, and the other half elements in the multi-set have multiplicity $\frac{q-1}{2}$.
          So the hypergraph has
		\[
		|V|= \frac{\frac{\frac{q^2-q}{2}(q+1)}{2}}{\frac{q+1}{2}}+\frac{\frac{\frac{q^2-q}{2}(q+1)}{2}}{\frac{q-1}{2}}=q^2
		\]
		vertices.
	\end{proof}
	
	With the preparation above, we now turn to our main result of this section.
	\begin{thm}
		Let $q$ be an odd prime power. When $p(x)$ runs over all the monic irreducible polynomials of degree $2$ over $\f{q}$, the deep holes in Theorem~\ref{deephole:irr2} are all the deep holes of $PRS(q+1,q-2)$.
	\end{thm}
	\begin{proof}
		By Theorem~\ref{numberofvertices}, there are
		\[
		|V|(q-1)= q^2(q-1)
		\]
	cosets of deep holes found in Theorem~\ref{deephole:irr2}.
	
	On the other hand, by Theorem~\ref{numberofdeepholes:k=q-2} there are totally $q^2(q-1)$ cosets of deep holes. So the deep holes in Theorem~\ref{deephole:irr2} are all the deep holes of $PRS(q+1,q-2)$ if $p(x)$ runs over all the monic irreducible polynomials of degree $2$.

\end{proof}

\section{Deep holes of $PRS(q+1,q-3)$}\label{section:k=q-3}
In this section, we extend the construction using irreducible quadratic polynomials to irreducible cubic polynomials for the code $PRS(q+1,q-3)$. More deep holes of $PRS(q+1,q-3)$ are obtained.

In Section~\ref{section:k=q-2} it is proved that for odd $q$ and $k=q-2$ there are totally $q^2(q-1)$ cosets of deep holes constructed in Theorem~\ref{deephole:irr2} which form all the deep holes of the PRS code $PRS(q+1,q-2)$. But for the code $PRS(q+1, q-3)$, by Lemma~\ref{disjoint:degree2}
the total number of cosets of deep holes in Theorems~\ref{deephole:irr2} and~\ref{deephole:degreek} is
\[
\frac{q^2-q}{2}(q^2-1)+q(q-1)=(q-1)\left(\frac{q^3}{2}+\frac{q}{2}\right).
\]
By Theorem~\ref{numberofdeepholes:k=q-3} the exact number of cosets of deep holes of $PRS(q+1, q-3)$ is
\[
(q-1)\left(\frac{q^3}{2}+q^2+\frac{q}{2}\right).
\]
Hence there are other deep holes! 

\begin{thm}[Construction from irreducible cubic polynomials]\label{deephole:irr3}
	Let $q$ be a prime power and let $k=q-3$. Let $C=PRS(q+1,k)$. Let $p(x)\in\f{q}[x]$ be a monic irreducible polynomial of degree $3$.
	\begin{enumerate}
		\item There are totally
		\[
		(q-1)\frac{q^2+q+2}{2}
		\]
		cosets of deep holes of the form
		\[
		(u_f,0)+C
		\]
		where $f(x)=\frac{a+bx+cx^2}{p(x)}$ for some non-zero polynomial $a+bx+cx^2\in\f{q}[x]$.
		\item There are new cosets of deep holes of $C$ of the form above besides the deep holes in Theorem~\ref{deephole:irr2} and Theorem~\ref{deephole:degreek}.
	\end{enumerate}
	
\end{thm}
\begin{proof}
	We have seen that in this case the covering radius $\rho({PRS(q+1,k)})= q-k$. Similarly to Theorem~\ref{deephole:irr2}, we construct a supercode $C'$ using the irreducible polynomial $p(x)$. Then $C'$ is an MDS code with parameters $[q+1,k+3]$. So for any rational function $f(x)=\frac{a+bx+cx^2}{p(x)}$ where $a,b,c$ are not zero simultaneously,
	we have
	\[
	q-k-1= d(C')\leq d((u_f,0),C)\leq\rho({PRS(q+1,k)})= q-k.
	\]
	The error distance $d((u_f,0),C)= q-k-1$ if and only if there is $h(x)\in\f{q}[x]$ of degree $ k-2$ such that
	\[
	a+bx+cx^2-h(x)p(x)
	\]
	splits as $k+1=q-2$ linear factors, or
	\[
	a+bx+cx^2-(dx^{k-1}+h(x))p(x)
	\]
	splits as $k+2=q-1$ linear factors for some $d\in\f{q}^*$.
	Moreover, for fixed $a+bx+cx^2$ if there is $h(x)\in\f{q}[x]$ of degree $ k-2$ such that
	\[
	a+bx+cx^2-h(x)p(x)
	\]
	splits as $k+1=q-2$ linear factors, then such $h(x)$ is unique. If not, suppose there is another $h_1(x)\in\f{q}[x]$ of degree $ k-2$ such that
	\[
	a+bx+cx^2-h(x)p(x)
	\]
	splits as $k+1=q-2$ linear factors. Let $S$ and $T$ be ($q-2$)-subsets of $\f{q}$ such that
	\[
	a+bx+cx^2-h(x)p(x)=d\prod_{s\in S}(x-s)
	\]
	and
	\[
	a+bx+cx^2-h_1(x)p(x)=d_1\prod_{t\in T}(x-t)
	\]
	for some $d,d_1\in\f{q}^*$. Then
	\[
	(h(x)-h_1(x))p(x)=d_1\prod_{t\in T}(x-t)-d\prod_{s\in S}(x-s).
	\]
	Note that
	\[
	|S\cap T|\geq q-4.
	\]
	So the polynomial $(h(x)-h_1(x))p(x)$ has at least $q-4$ linear factors. Since the polynomial $h(x)-h_1(x)$ has degree
	\[
	\deg(h(x)-h_1(x))\leq q-5,
	\]
	the polynomial $p(x)$ has at least $1$ linear factor, which contradicts to the irreducibility of $p(x)$. So for $a+bx+cx^2$ if there is $h(x)\in\f{q}[x]$ of degree $ k-2$ such that
	$a+bx+cx^2-h(x)p(x)$
	splits as $k+1=q-2$ linear factors, then the factorization is unique. On the other hand, any $(q-2)$-subset $S$ of $\f{q}$ contributes $q-1$ polynomials $a+bx+cx^2$ in the following way
	\[
	a+bx+cx^2\equiv d\prod_{s\in S}(x-s) \mod p(x)
	\]
	where $d\in\f{q}^*$ is arbitrary.
	 Hence, the number of such $a+bx+cx^2$ is
	\[
	(q-1)\binom{q}{q-2}.
	\]
	Similarly, the number of polynomials $a+bx+cx^2$ such that
	$a+bx+cx^2-h(x)p(x)$
	splits as $k+2=q-1$ linear factors for some $h(x)$ of degree $k-1$ is
	\[
	(q-1)\binom{q}{q-1}.
	\]
	By the same argument again, there is no $a+bx+cx^2$ satisfying the two splitting conditions at the same time.
	So the total number of $f$ satisfying  $d((u_f,0),C)= q-k-1$ is
	\[
	(q-1)\left(\binom{q}{q-2}+\binom{q}{q-1}\right).
	\]
	Hence the number of cosets of deep holes of $C$ of the form $ \left(u_{\frac{a+bx+cx^2}{p(x)}},0\right)+C$ is
	\[
	q^3-1- (q-1)\left(\binom{q}{q-2}+\binom{q}{q-1}\right)=(q-1)\frac{q^2+q+2}{2}.
	\]
	This proves the first statement.
	
	Next we show that there are NEW deep holes besides the deep holes in Theorems~\ref{deephole:irr2} and~\ref{deephole:degreek}.
	
	FIX $p(x)=x^3+\alpha x^2+\beta x+\gamma$ a monic irreducible polynomial. Suppose there are a monic irreducible polynomial $q(x)$ of degree $2$ and a polynomial $f(x)\in \f{q}[x]$ of degree $\leq k-2$ such that
	\[
	\frac{a+bx+cx^2}{p(x)}=\frac{d+ex}{q(x)}+f(x)\qquad\forall x\in \f{q}.
	\]
	Then one can solve to get
	\[
	a+bx+cx^2\equiv \frac{l(x^q-x)}{q(x)}\mod p(x)
	\]
	and
	\[
	d+ex \equiv \frac{-l(x^q-x)}{p(x)}\mod q(x)
	\]
	for some $l\in\f{q}^*$. So running over all monic irreducible polynomial of degree $2$ and $l\in\f{q}^*$, they contribute at most
	\[
	\frac{q^2-q}{2}(q-1)
	\]
	cosets of deep holes of the form in the lemma. So besides the deep holes in Theorem~\ref{deephole:irr2}, there are at least
	\[
	(q-1)\frac{q^2+q+2}{2}-\frac{q^2-q}{2}(q-1)=(q-1)(q+1)
	\]
	cosets of deep holes of the form $\left(u_{\frac{a+bx+cx^2}{p(x)}},0\right)+C$.
	
	We have already seen that the number\footnote{Indeed, for fixed $p(x)=x^3+\alpha x^2+\beta x+\gamma$ one can show only $\left(u_{ex^k-e\alpha x^{k-1}},0\right)+C,\,e\in\f{q}^*$ can be represented by $\left(u_{\frac{a+bx+cx^2}{p(x)}},0\right)+C$ where $a+bx+cx^2\equiv -e(x^q-x)\mod p(x)$.} of cosets of deep holes in Theorem~\ref{deephole:degreek}
	\[
	|\{(u_f,v)+C\,|\,\deg(f)=k,\,v\in\f{q}\}|= q(q-1).
	\]
	So there are at least
	\[
	(q-1)(q+1)-(q-1)q=q-1
	\]
	cosets of deep holes of the form $\left(u_{\frac{a+bx+cx^2}{p(x)}},0\right)+C$ besides the deep holes in Theorem~\ref{deephole:irr2} and Theorem~\ref{deephole:degreek}.

\end{proof}

Next we study the relationship between deep holes in Theorem~\ref{deephole:irr3} and the deep holes in Theorems~\ref{deephole:irr2} for $k=q-3$. This problem is related to the distribution problem of irreducible polynomials which has been studied extensively in number theory (see~\cite{Wan97} and~\cite{Coh05}). However, the lower bound in the references is vacuous in our situation. To overcome the difficulty, we use the method in~\cite{Wan97} to give an explicit formula for the distribution of irreducible cubic polynomials.
\begin{lem}[Distribution of irreducible cubic polynomials]\label{distributionofcubicirrepoly1}
	For any non-zero polynomial $a+bx\in\f{q}[x]$ and any monic irreducible polynomial $q(x)\in\f{q}[x]$ of degree $2$, there exists a monic irreducible polynomial $p(x)\in\f{q}[x]$ of degree $3$ such that
	\[
	p(x)  \equiv l(a+bx)\mod q(x)
	\]
	for some $l\in\f{q}^*$.
\end{lem}

\begin{proof} Let $G = (\f{q}[x]/(q(x)))^*$. This is just the multiplicative group of the quadratic extension of the field $\f{q}$. For $\alpha \in G$, let
	$N_3(\alpha)$ denote the number of pairs $(p(x), l)$ such that $p(x)\in \f{q}[x]$ is monic irreducible of degree $3$, $l\in \f{q}^*$ and
	$p(x) \equiv l \alpha \mod q(x)$. We want to show that $N_3(\alpha)>0$. In fact, we will give a simple explicit formula for $N_3(\alpha)$
	which immediately imply that $N_3(\alpha)>0$.
	
	Let $\hat{G}$ be the complex character group of $G$, which is the cyclic group of order $q^2-1$. Let $\hat{G_1}$ denote the subgroup of $\hat{G}$
	consisting of those characters whose restriction to $\f{q}^*$ is trivial. It is clear that $\hat{G_1}$ is a cyclic group of order $q+1$. Let $\pi_k$
	denote the set of monic irreducible polynomials of degree $k$ in $\f{q}[x]$. A standard character sum argument shows that
	$$N_3(\alpha) = \frac{1}{q^2-1} \sum_{\chi \in \hat{G}} \sum_{p(x)\in \pi_3} \sum_{l \in \f{q}^*} \chi(\frac{p(x)}{l\alpha}).$$
	If $\chi \not\in \hat{G_1}$, then $\sum_{l\in \f{q}^*}\chi(l^{-1}) =0$. It follows that
	\begin{align*}
	N_3(\alpha) &= \frac{q-1}{q^2-1} \sum_{\chi \in \hat{G_1}} \sum_{p(x)\in \pi_3} \chi(\frac{p(x)}{\alpha})\\
	&=\frac{1}{q+1}\left(\frac{q^3-q}{3}
	+\sum_{\chi \in \hat{G_1}\setminus\{1\}} \chi^{-1}(\alpha) \sum_{p(x)\in \pi_3} \chi(p(x))\right).
	\end{align*}

	Now $\chi$ is a nontrivial character of $G$ which is trivial on $\f{q}^*$. For positive integers $k$,  consider the character sums and their L-function
	$$S_k(\chi) = \sum_{d|k} \sum_{p(x) \in \pi_d} d\chi(p(x)^{k/d}), \ \ L(\chi, T) = \exp(\sum_{k=1}^{\infty}\frac{S_k(\chi)}{k}T^k).$$
	By \cite{Wan97}, the L-function $L(\chi, T)$ of $\chi$ is a
	polynomial of degree $1$ with the trivial factor $1-T$.  This can also be proved directly. It follows that $L(\chi, T) = 1-T$ and $S_k(\chi) =-1$ for all $k$.  	
	Taking $k=3$, we obtain
	$$-1 = S_3(\chi) = 3\sum_{p(x) \in \pi_3} \chi(p(x)) + S_1(\chi^3).$$
	If $\chi^3 \not =1$, then $S_1(\chi^3) = -1$ and we deduce
	$$\sum_{p(x) \in \pi_3} \chi(p(x)) =0$$
	If $\chi^3=1$, then $\chi=\chi_3$ is a primitive character of order $3$ trivial on $\f{q}$ (necessarily $q\equiv 2 \mod 3$), 	then $S_1(\chi^3) =q$
	and
	$$\sum_{p(x) \in \pi_3} \chi(p(x)) = -\frac{q+1}{3}.$$	
	Putting the above together, we deduce the following result.
	If $q\not\equiv 2\mod 3$, then $N_3(\alpha) = q(q-1)/3$. If $q\equiv 2 \mod 3$, then
	$$N_3(\alpha) = \frac{1}{3}( q(q-1) - (\chi_3(\alpha) +\chi_3^{-1}(\alpha))).$$
	In summary,  we have
	$$N_3(\alpha) = \frac{1}{3} (q(q-1) -r_3(\alpha)),$$
	where $r_3(\alpha) =0$ if $q\not\equiv 2 \mod 3$, $r_3(\alpha) = -1$ if $q\equiv 2 \mod 3$ and $\chi_3(\alpha)\not=1$, and 	
	$r_3(\alpha) = 2$ if $q\equiv 2 \mod 3$ and $\chi_3(\alpha)=1$.
\end{proof}

\begin{thm}
	Let $q$ be a prime power and let $k=q-3$. The deep holes in Theorem~\ref{deephole:irr2} are contained in deep holes in Theorem~\ref{deephole:irr3}.
\end{thm}
\begin{proof}
	 Let $C=PRS(q+1,k)$. For any coset of deep holes in Theorem~\ref{deephole:irr2}
	\[
	\left(u_{\frac{d+ex}{q(x)}},0\right)+C
	\]
	where $q(x)$ is a monic irreducible quadratic polynomial and $d+ex$ is non-zero,
	by Lemma~\ref{distributionofcubicirrepoly1} there is a monic irreducible polynomial $p(x)$ of degree $3$ such that
	\[
	p(x)  \equiv \frac{l(x^q-x)}{d+ex}\mod q(x)
	\]
	for some $l\in\f{q}^*$. Let
	\[
	a+bx+cx^2\equiv \frac{-l(x^q-x)}{q(x)}\mod p(x).
	\]
	Then one can show that there is a polynomial $f(x)\in \f{q}[x]$ of degree $\leq k-2$ such that
	\[
	\frac{a+bx+cx^2}{p(x)}=\frac{d+ex}{q(x)}+f(x)\qquad\forall x\in \f{q}.
	\]
	So
	\[
	\left(u_{\frac{d+ex}{q(x)}},0\right)+C=	\left(u_{\frac{a+bx+cx^2}{p(x)}},0\right)+C.
	\]
\end{proof}

The theorem below shows that the deep holes constructed from monic irreducible cubic polynomials contain the deep holes of degree $k$.
\begin{thm}
	Let $q$ be a prime power and let $k=q-3$. The deep holes in Theorem~\ref{deephole:degreek} are contained in deep holes in Theorem~\ref{deephole:irr3}.
\end{thm}

\begin{proof}
	 Let $C=PRS(q+1,k)$. For any coset of deep holes $(u_f,v)+C$ in Theorem~\ref{deephole:degreek}, wlog, we may assume $f(x)=x^k-wx^{k-1}$, as $(u_{x^i},0)\in PRS(q+1,k)$ for all $0\leq i\leq k-2$ and any non-zero scale multiple of a deep hole is still a deep hole. Let
	\[
	g(x)=f(x)-vx^{k-1}=x^k-(w+v)x^{k-1}.
	\]
	Then we have
	\[
	d((u_g,0),PRS(q+1,k))=d((u_f,v),PRS(q+1,k)).
	\]
	
	By theorems in~\cite{Wan97} and~\cite{Ham98} (also see~\cite[Theorem~2.7]{Coh05}), there exists a monic irreducible polynomial $p(x)\in \f{q}[x]$ with $x^2$-coefficient $w+v$. Since $x^q-x$ is coprime to $p(x)$, we have
	\[
	a+bx+cx^2\equiv -(x^q-x)\mod p(x)
	\]
	is non-zero. So there exists some polynomial $h(x)\in\f{q}[x]$ of degree $\leq k-2$ such that
	\[
	(g(x)-h(x))p(x)={a+bx+cx^2}+x^q-x,
	\]
	which is also equivalent to
	\[
	g(x)=\frac{a+bx+cx^2}{p(x)}+h(x)\qquad \forall x\in \f{q}.
	\]
	So any coset in $\{(u_f,v)+C\,|\,\deg(f)=k,\,v\in\f{q}\}$ can be written as
	\[
	\left(u_{\frac{a+bx+cx^2}{p(x)}},0\right)+C
	\]
	for some monic irreducible polynomial $p(x)\in\f{q}[x]$ of degree $3$ and some non-zero polynomial $a+bx+cx^2\in\f{q}[x]$.
	
\end{proof}

	\begin{rem}
	\begin{enumerate}
		\item It is natural to ask that in the case $k=q-3$ if deep holes constructed from all irreducible polynomials of degree $3$ are all the deep holes of $PRS(q+1,q-3)$. For small $q$, using mathematical software, e.g., SageMath~\cite{sagemath}, one could count the number of cosets of deep holes constructed from irreducible polynomials of degree $3$ and verify if it equals $(q-1)\left(\frac{q^3}{2}+q^2+\frac{q}{2}\right)$. The experiment result would suggest that the deep holes constructed from irreducible polynomials of degree $3$ are not complete for prime power $q\geq 8$.
		\item An open problem is to count the number of cosets of deep holes constructed in Theorem~\ref{deephole:irr3} when the construction runs over all the irreducible polynomials of degree $3$ and find all the deep holes of $PRS(q+1,q-3)$.
	\end{enumerate}
	
\end{rem}

\section*{Acknowledgement}
Part of this paper was written when the two authors were visiting Beijing International Center for Mathematical Research (BICMR). The authors would like to thank Prof. Ruochuan Liu for his hospitality. Part of this paper was written when the first author was visiting University of Oklahoma. He would like to thank Prof. Qi Cheng for his hospitality.

\bibliographystyle{IEEETranS}
\bibliography{covering}
\end{document}